\documentclass[showpacs,pra,superscriptaddress,amsmath,%
amsfonts,amssymb,a4paper,twocolumn]{revtex4}

\usepackage[english]{babel}

\usepackage{amsthm}

\DeclareMathOperator{\tr}{Tr}

\newtheorem{cor}{Corollary}

\newtheorem{prop}{Proposition}

\usepackage{euscript,amssymb,amsbsy}

\newcommand{\be}[1]{\begin{equation}\label{#1}}
\newcommand{\ee}{\end{equation}}
\newcommand{\ba}[1]{\begin{eqnarray}\label{#1}}
\newcommand{\ea}{\end{eqnarray}}
\newcommand{\rf}[1]{(\ref{#1})}
\newcommand{\nn}{\nonumber}

\def\a{\alpha}
\def\b{\beta}
\def\d{\delta}
\def\g{\gamma}
\def\G{\Gamma}
\def\sg{\sigma}

\def\o{\omega}

\def\ra{\rangle}
\def\la{\langle}
\def\dg{\dagger}

\newcommand{\wt}{\widetilde}

\begin{document}

\title{Minimum error discrimination problem for pure qubit states}

\author{Boris F Samsonov}
\affiliation{
 Physics Department, Tomsk State University, 36 Lenin Avenue,
634050 Tomsk, Russia}

\begin{abstract}
The necessary and sufficient conditions for minimization of
the generalized rate error for discriminating among $N$ pure
qubit states are reformulated in terms of Bloch vectors
representing the states.
For the direct optimization problem
an algorithmic solution to these
conditions is indicated.
A solution to the inverse optimization problem is given.
General results are widely illustrated by particular cases
of equiprobable states and $N=2,3,4$ pure qubit states
given with different prior probabilities.
\end{abstract}

\maketitle

\section{Introduction}

Essential advances in experimental techniques in recent
years made possible high precision measurements which are
required for distinguishing among separate
nonorthogonal quantum states.
These studies are stimulated by needs of quantum
communication and quantum cryptography where the
discrimination of quantum states is one of the key
problems.
In a more general context, this problem underlies
many of the communication and computing
schemes that have been suggested so far \cite{BHH}.
According to quantum-mechanical laws there is no way
to discriminate perfectly among nonorthogonal states.
Therefore, an actual problem is the state discrimination
in an optimal way.
Currently, different criteria of optimality are formulated
(for a recent review, see \cite{Barnett-Croke}).
In this paper we will consider discrimination strategies
 based on the minimization of the generalized rate error
(or more generally, the mean Bayes' cost \cite{Helstr,Chefless}).
They are known
in the literature as the
{\em minimum-error discrimination} strategies.

Necessary and sufficient conditions for
realizing  the minimum-error (generalized)
measurement were formulated independently
 by Holevo \cite{Hol,Hol2} and
 Yuen {\em et al}
\cite{Yuen} (see also a recent appealingly simple proof in
\cite{Barnett-Croke-JPA}).
Unfortunately, as Hunter \cite{Hunter} pointed out,
these conditions do
not provide a great insight into either the form of
minimum-error measurement strategies, or into how error probability
depends on the set of possible states.
Moreover, the only strategy derived from these conditions
corresponds to two possible states \cite{Helstr}.
All other optimal solutions, except the one proposed by Hunter
\cite{Hunter} for equiprobable qubit states,
were first postulated and then shown to
satisfy these conditions.
Actually, Hunter's solution \cite{Hunter} is based on
solving these conditions only partly and
consists in the formulation of a two-step procedure.
First, he proposed to find some auxiliary operators
and then to check which of these
operators correspond to possible measurements.
In this way, he was able to find a solution for equiprobable
qubit states with a clear geometrical visualization.
We have to note that our approach
 based on the full solution of the necessary
and sufficient conditions
permitted us to draw slightly different conclusions
compared to these published in \cite{Hunter}.

In the current paper,
we first formulate conditions,
to which the Lagrange
operator should satisfy,
when a measurement strategy remains to be optimal after a
given set of states is enlarged by new states.
Then
we reformulate the known necessary and
sufficient optimization conditions
 for pure qubit states given with
arbitrary prior probabilities
in terms of Bloch vectors representing the states.
This permitted us to indicate an algorithmic solution of the
(direct) optimization problem.
Furthermore, using the new form of the necessary and sufficient
optimization, conditions we succeeded to solve the
{\em inverse optimization problem},
i.e., the problem of finding all sets
of states and prior probabilities for which a given
measurement strategy is optimal.
For two-qubit states, we found a solution to the
optimization problem for the case when the states are given
with probabilities $p_1$ and $p_2$ such that $p_1+p_2\le1$.
For $p_1+p_2=1$, our
solution reduces to the one first found by
 Helstrom \cite{Helstr}
 (see also \cite{Barnett-Croke}).
Using this result, we derived a condition for $N$ qubit
states $\rho_1$, \ldots, $\rho_N$ when they are optimally
 discriminated by the
same orthogonal measurement which is optimal for
discriminating among the states $\rho_1$ and $\rho_2$.
This permitted us to formulate conditions
for distinguishing projective measurements
from generalized ones for optimal discrimination among
three qubit states.
For $N$ equiprobable states, we have shown that the optimal
(generalized) measurement may contain $M\ge4$
elements of a positive operator-valued measure
(POVM)
 only if there
exists a subset of $M$
states which may form a POVM.
Otherwise, the optimal POVM
contains either three or two elements.
We applied our algorithm to the case of three states, two of
which are given with equal probabilities.
For a particular
case of three mirror symmetric states, our result coincides
with the one previously published
by Anderson et al. \cite{Anderson}.
For
another particular case, we compared our solution to a
solution obtained numerically in \cite{Czhec}. In contrast
to \cite{Czhec}, our approach gives exact analytic bounds
for the prior probability, distinguishing projective optimal
measurement from the generalized one.
As another illustration of our approach, we obtained the
optimal solution for four qubit states, two of which are
given with the probability $p$ and two others with
the probability $1/2-p$.

\section{Minimum error conditions for pure qubit states}

\subsection{Minimum error conditions \label{MEC}}

Assume that we are given $N(\ge2)$ pure qubit
states
\be{1}
\rho_j=\rho_j^\dg=\rho_j^2\ge0\,, \quad \tr\rho_j=1\,,\quad
 j=1,\ldots,N
\ee
and a (discrete) probability
distribution $p_j>0$, $\sum_{j=1}^Np_j=1$.
We assume also that the
operator $\rho_j$
occurs with the probability $p_j$ and
acts on the vectors of a two-dimensional
Hilbert space with the usually defined inner product.
A (generalized, see, e.g., \cite{GM}) measurement that
discriminates between the states $\rho_1$,\ldots,$\rho_N$
can be described with the help of $N$ detection operators
$\hat\pi_1$,\ldots,$\hat\pi_N$ such that
\ba{pih}
\hat\pi_j^\dg=\hat\pi_j\ge 0\,,
\\
\sum_{j=1}^N\hat\pi_j=I& &
\label{I}
\ea
 where $I$ is the identity operator.
 If as a measurement result the detector $\pi_{i_0}$
 clicks we know
 with the
 probability
$\tr(\rho_{i_0}\hat\pi_{i_0})$
that
 the chosen state is $\rho_{i_0}$.
The overall probability of
correctly identifying any of the states $\rho_j$
is then given by (see, e.g., \cite{Barnett-Croke})
\be{Pcorr}
P_{corr}=\sum_{i=1}^Np_i\tr(\hat\pi_i\rho_i)
=\tr\Gamma
\ee
where we introduced
\be{Gam}
\Gamma:=\sum_{j=1}^Np_j\hat\pi_j\rho_j
\ee
and $P_{err}=1-P_{corr}$ is
the overall probability of an erroneous guess.
The minimum error discrimination strategy corresponds to
such a set of detection operators $\hat\pi_j$ that
$P_{err}$ takes its minimal value or equivalently
$P_{corr}$ takes its maximal value.
Operator $\Gamma$ \rf{Gam} is known in the literature as
the {\em Lagrange operator} (see e.g. \cite{Chefless}).
Since its trace defines $P_{corr}$ we will call its matrix
representation in a basis the {\it cost matrix}.

According to \cite{Hol,Yuen},
the necessary and sufficient conditions for a measurement
to be optimal are
\ba{cond1}
\Gamma=\Gamma^\dg\,, & & \\
G_j:=\Gamma&-&p_j\rho_j\ge0\,,\quad  j=1,\ldots,N\,.
\label{cond2}
\ea
As it was proved by Holevo \cite{Hol2}, instead of
$\Gamma$,
one can equivalently use $\wt\Gamma=(\Gamma+\Gamma^\dg)/2$
and at the extremum point $\wt\Gamma=\Gamma$.
Below, we will use also
the following implication of Eq. \rf{cond2}
\cite{Hol2,Hunter,Barnett-Croke}
\be{NC}
G_j\hat\pi_j=0\,,\quad j=1,\ldots,N\,.
\ee
It is not difficult to see that from Eq. \rf{NC}, it follows
that for qubit states,
 if both $\hat\pi_j\ne0$ and $G_j\ne0$, then they both have
a zero eigenvalue and therefore they both are proportional
to projectors (cf. \cite{Hunter}).
Indeed, since $\hat\pi_j$ is Hermitian one can choose a
coordinate system where it is diagonal with at least one
non-zero diagonal entry.
In this system from
\rf{NC}, it follows that at least one column of $G_j$
and one diagonal element of $\hat\pi_j$ should be equal to
zero, i.e., both $\det\hat\pi_j=0$ and $\det G_j=0$
and both these operators are proportional to projectors.
In what follows, we will extensively use this property.

We would like to note here that
for an optimal strategy, some of the
operators $\hat\pi_j$ may become equal to zero.
This
happens if some of the detectors $\hat\pi_j$,
say for definiteness the detectors with the numbers $j=M+1,\ldots,N$,
never click.
In this case,
only the states $\rho_j$, with $j=1,\ldots,M$, enter
the cost matrix $\Gamma$ \rf{Gam} and their probabilities satisfy
the condition $\sum_{j=1}^Mp_j<1$
while the first $M$ POVM elements, $\hat\pi_j$,
$j=1,\ldots,M$, satisfy conditions \rf{pih} and \rf{I} with
$N=M$.
We would like to stress that this property may be formulated in a
more constructive way.
\begin{prop}\label{prop0}
Assume that for $M$ states $\rho_j$,
$j=1,\ldots,M$,
given with the probabilities $p_j$
such that
$\sum_{j=1}^Mp_j<1$,
the detection operators $\hat\pi_j$,
satisfying conditions (\ref{pih}) and (\ref{I}) at $N=M$,
which maximizes the success probability $P_{corr}$
\rf{Pcorr}
for
 the cost matrix $\Gamma$ \rf{Gam} at $N=M$,
are known.
Then for a larger set of $N>M$
states
$\wt\rho_j$, $j=1,\ldots,N$
given with the probabilities $\wt p_j$,
$\sum_{j=1}^N\wt p_j=1$,
such that $\wt\rho_j=\rho_j$
and $\wt p_j=p_j$, $j=1,\ldots,M$,
the optimal POVM contains all previous elements
$\hat\pi_j$, $j=1,\ldots,M$
plus the zero elements
$\hat\pi_j=0$, $j=M+1,\ldots,N$
provided
\be{newSet}
\Gamma-\wt p_j\wt\rho_j\ge0\,,\quad j=M+1,\ldots,N\,.
\ee
The optimal success probability for the new set is
$\wt P_{corr}=P_{corr}$.
\end{prop}

As the final comment of this section, we note that in what
follows, we will distinguish the projective
(i.e., von Neumann) measurement from the generalized one.
The projective measurement may be realized only if there
are no more than two detection operators in which case they
should satisfy the usual condition $\hat\pi_j^2=\hat\pi_j$.
For the generalized measurement one needs more than two
detection operators.

\subsection{Minimum error conditions in terms
of Bloch vectors\label{MinSubs}}

We find convenient to choose
Pauli matrices $\sg_i$ ($i=1,2,3$) and
 $\sg_0=I$  as a basis
 in the linear space of Hermitian $2\times2$
matrices
\ba{PailiM}\nn
\sg_j\sg_k=-\sg_k\sg_j=i\sg_l\,,\quad
\sg_j^2=\sg_0=I\,,
\\
\tr\sg_j=0\,,\quad \tr(\sg_j\sg_k)=2\d_{j,k}\,,
\quad j,k,l=1,2,3\,.
\ea
Thus, once $N$ states $\rho_j$ are given, the coefficients
$\b_{j,k}$ defining the states
\be{rho}
\rho_j=\sum_{k=0}^3\b_{j,k}\sg_k\,,\quad
j=1,\ldots,N
\ee
may be found as
$\b_{j,k}=\frac{1}{2}\tr(\rho_j\sg_k)\in\mathbb R$.
Normalization condition \rf{1}
imposes the following restrictions on $\b_{j,k}$:
\be{bt0}
\b_{j,0}=\frac12\,,\quad\sum_{k=1}^3\b_{j,k}^2=\frac14\,.
\ee
Therefore, any state is defined by a
point on a sphere with the
radius equal to $1/2$ or equivalently by a
$3$-dimensional real vector (Bloch vector)
$\vec{\b}_j=(\b_{j,1},\b_{j,2},\b_{j,3})$
of the length equal to $1/2$.

As it was argued in the previous section
 we may look for the detection operators in
the form
\be{pij}
\hat\pi_j=\o_j\pi_j\,,\quad
\pi_j^\dg=\pi_j=\pi_j^2>0\,,\quad \o_j\ge0\,.
\ee
We shall call $\o_j$ the {\em frequencies}.
Taking into account Proposition \ref{prop0},
 we may look for
$M$ operators $\pi_j$ optimizing the
measurement strategy for $M\le N$ states $\rho_j$,
$j=1,\ldots,M$,
given with the probabilities $p_j$,
$\sum_{j=1}^Mp_j\le1$.
The problem is solved if such a strategy is found and the
remaining states (if $M<N$) satisfy conditions
\rf{newSet}.
Note that the optimization problem has always a solution
(see \cite{Hol2,Yuen}).

Operators
$\pi_j$ will be defined in terms of unknown coefficients
$\g_{j,k}$,
\be{pi}
\pi_j=\sum_{k=0}^3\g_{j,k}\sg_k
\ee
which should satisfy the relations similar to Eq. \rf{bt0}
\be{gjk}
\g_{j,0}=\frac12\,,\quad\sum_{k=1}^3\g_{j,k}^2=\frac14\,,
\ee
and every $\rho_j$ is also defined by a $3$-dimensional
real vector
$\vec\g_j=(\g_{j,1},\g_{j,2},\g_{j,3})$ of the length equal to $1/2$.
Moreover, from Eq. \rf{I} one finds additional restrictions
\ba{om1}
\sum_{j=1}^M\o_j\g_{j,k}&=&0\,,
\quad k=1,2,3\,,
\\
\sum_{j=1}^M\o_j&=&2\,,\quad \o_j\ge0\,.
\label{om2}
\ea
According to Eq. \rf{Gam} the cost matrix
is also defined in terms of the vectors $\vec\b_j$
and $\vec\g_j$,
\ba{GS}
\Gamma&=&\sum_{k,k'=1}^3\Gamma_{k,k'}\sg_k\sg_{k'},\\
\Gamma_{k,k'}&=&\sum_{j=1}^M\o_jp_j\b_{j,k}\g_{j,k'}\,,
\label{Gkk}
\ea
with the vectors $\vec\g_j$ such that
$\sum_{j=1}^M\o_j\vec\g_j=0$.

Now we can prove the following statement.
\begin{prop}\label{prop1}
If $\o_j>0$, $j=1,\ldots,M$, then
conditions \rf{cond2} are satisfied for
$\hat\pi_j=\o_j\pi_j$, where $\pi_j$
are given by Eq. \rf{pi}
with
\be{gamma}
\g_{j,k}=\frac{\b_{j,k}p_j-B_k}{A-p_j}\,,
\quad
k=1,2,3,\ j=1,\ldots,M\,,
\ee
where
\be{Aprop}
A=\tr\Gamma>p_j\,,\quad\forall j
\ee
 and $B_k\in\mathbb R$ are arbitrary parameters.
\end{prop}

\begin{proof}
According to Eqs. \rf{GS} and \rf{PailiM}, one may write
\be{Ga}
\Gamma=\sum_{k=0}^3\Gamma_{k,k}\sg_0+
\sum_{k=1}^3(\Gamma_{k,0}+\Gamma_{0,k})\sg_k+
\sum_{k\ne k'=1}^3\Gamma_{k,k'}\sg_{k}\sg_{k'}\,,
\ee
where [see Eqs. \rf{Gkk}, \rf{bt0} and \rf{gjk}]
\ba{Gk0}
\Gamma_{k,0}&=&\frac12\sum_{j=1}^M\o_jp_j\b_{j,k}\,,\\
\Gamma_{0,k}&=&\frac12\sum_{j=1}^M\o_jp_j\g_{j,k}\,.\label{G0k}
\ea
As it was mentioned in Sec. \ref{MEC}, instead of
$\Gamma$, one can equivalently use
$\wt\Gamma=(\Gamma+\Gamma^\dg)/2$ which in our case is
given by the first two terms at the right-hand side of
Eq. \rf{Ga}.
Therefore, below, to simplify notations, we will mean
by $\Gamma$ the same expression \rf{Ga} where the last
term is absent.
This is justified by the property which we will show below.
Namely, we will show that from Eq. \rf{NC}
follows the Hermitian character of $\Gamma$ and hence the
zero contribution to $\Gamma$ from the last term in Eq.
\rf{Ga}.
This leads to the following expression for $G_j$:
(\ref{cond2})
\be{GPI}
G_j=\left(\sum_{k=0}^3\Gamma_{k,k}-\frac12p_j\right)\sg_0+
\sum_{k=1}^3(\Gamma_{k,0}+\Gamma_{0,k}-p_j\b_{j,k})\sg_k\,.
\ee
Using Eqs. \rf{NC}, and \rf{GPI}
and assumed property $\o_j\ne0$, i.e., $G_j\pi_j=0$,
we obtain a set of equations for coefficients $\g_{j,k}$,
\be{set1}
\sum_{k=0}^3\Gamma_{k,k}-\frac12p_j+
2\sum_{k=1}^3(\Gamma_{k,0}+\Gamma_{0,k}-p_j\b_{j,k})\g_{j,k}=0\,,
\ee
\be{set2}
\Gamma_{k,0}+\Gamma_{0,k}-p_j(\b_{j,k}+\g_{j,k})+
2\g_{j,k}\sum_{k'=0}^3\Gamma_{k',k'}=0\,,
\ee
\be{set3}
(\Gamma_{k,0}+\Gamma_{0,k}-p_j\b_{j,k})\g_{j,k'}=
(\Gamma_{k',0}+\Gamma_{0,k'}-p_j\b_{j,k'})\g_{j,k}\,.
\ee
If we multiply Eq. \rf{set1} by $\g_{j,k}$ and subtract the
result from Eq. \rf{set2}, we get
\ba{z1}
(\Gamma_{k,0}+\Gamma_{0,k}-p_j\b_{j,k})\g_{j,k'} & \nn \\
=4\sum_{k''=0}^3(\Gamma_{k'',0}+
\Gamma_{0,k''}-p_j\b_{j,k''})&\!\!\! \g_{j,k''}\g_{j,k}\g_{j,k'}\,.
\ea
The right- (and hence the left-) hand side of Eq. \rf{z1}
is symmetric with
respect to the permutation of $k$ and $k'$ and therefore
 Eq. \rf{set3} is an implication of Eqs. \rf{set1} and
 \rf{set2}.
From the other hand, multiplying Eq. \rf{set2} by $\g_{j,k}$
and summing up over $k$ gives just Eq. \rf{set1}.
From Eq. \rf{set2} follows the symmetry property
$\Gamma_{k,k'}=\Gamma_{k',k}$ and, hence, the Hermitian
character
of the matrix $\Gamma$.
Indeed, multiplying Eq. \rf{set2} by $\o_j\g_{j,k'}$,
summing up over $j$, and using Eqs. \rf{Gkk} and \rf{om2}
yields
\be{z2}
\Gamma_{k,k'}=-\sum_{j=1}^Mp_j\o_j\g_{j,k}\g_{j,k'}=
\Gamma_{k',k}\,,\quad k,k'=1,2,3\,.
\ee

Thus from the set of Eqs. \rf{set1}-\rf{set3}, we have
to solve Eq. \rf{set2} only.
For this purpose, we denote
\be{X}
X_{j,k}:=\Gamma_{k,0}+\Gamma_{0,k}-p_j(\b_{j,k}+\g_{j,k})\,.
\ee
These quantities have two remarkable properties.
The first property
\be{OX}
\sum_{j=1}^M\o_jX_{j,k}=0\,,\quad k=1,2,3
\ee
 follows from Eqs. \rf{om1}, \rf{Gk0}, and \rf{G0k}
 and the second property
 \be{OGX}
\sum_{j=1}^M\sum_{k=1}^3\o_j\g_{j,k}X_{j,k}=
-\sum_{k=0}^3\Gamma_{k,k}=
-\frac12\tr\Gamma
 \ee
 is a consequence of \rf{gjk}, \rf{om1}, \rf{GS}, and \rf{Gkk}.
 Using Eqs. \rf{X} and \rf{OGX}, we rewrite Eq. \rf{set2} as
 \be{eqX}
X_{j,k}=2\g_{j,k}\sum_{j'=1}^M\sum_{k'=1}^3\o_{j'}\g_{j'k'}X_{j'k'}\,,
 \ee
 which has a solution
\be{XA}
 X_{j,k}=A\g_{j,k}\,,
 \ee
 where $A$ is for the moment an
 arbitrary constant.
 Putting thus found $X_{j,k}$ to Eq. \rf{OGX}
 and using Eqs. \rf{gjk} and \rf{om2},
 we relate this constant with $P_{corr}$,
 $A=\tr\Gamma=P_{corr}$.

 To finish the proof of the statement, we solve Eq. \rf{X}
 with respect to $\g_{j,k}$.
 For this purpose, we replace
 in Eq. \rf{X} $\Gamma_{k,0}+\Gamma_{0,k}$ by
 the sum of Eqs. \rf{Gk0} and \rf{G0k}
 which yields
\[
X_{j,k}=\frac12\sum_{j'=1}^M\o_{j'}p_{j'}(\b_{j',k}+\g_{j',k})-
p_j(\b_{j,k}+\g_{j,k})\,.
\]
Then, using Eqs. \rf{om1} and \rf{om2},
 we find a solution to the above equation
 \be{bg}
p_j(\b_{j,k}+\g_{j,k})=-X_{j,k}+B_k\,,\quad
j=1,\ldots,M\,,\ k=1,2,3\,.
 \ee
Here, parameters $B_k$ remain arbitrary.
From here and Eq. \rf{XA}, we get Eq. \rf{gamma}.

Finally, we note
once again that since all $\o_j$ are assumed to be
different from zero, every matrix $G_j$ has a zero
eigenvalue (cf. \cite{Hunter}).
Therefore, the condition $G_j\ge0$ reduces to
$\tr G_j>0$ which upon using Eq. \rf{GPI} yields
$A>p_j$.
\end{proof}

Note that the proposition \ref{prop1} leaves parameter $A$
and the frequencies $\o_j$ unspecified.
Nevertheless, if we assume the frequencies to be invariant
under a scaling transformation of the probabilities,
i.e., after the replacement $p_j\to\a p_j$, one has
$\o_j\to\o_j$, then from Eq. \rf{Gkk}, it follows that $A$
becomes a
homogeneous function of the first degree with respect to
the probabilities $p_j$ provided the parameters $\g_{j,k}$
are also invariant with respect to the same scaling.
The last property agrees with Eq. \rf{gamma} from which the
scale invariance of the parameters $\g_{j,k}$ follows if
$A$ is a homogeneous function of the first degree with
respect to $p_j$. Note also that under these assumptions,
the detection operators $\hat\pi_j$ become invariant
with respect to the same scaling transformation.

\subsection{Direct optimization problem}

By the direct optimization problem, we mean the problem of
finding the operators $\hat\pi_j$ minimizing $P_{err}$ for
a given set of operators $\rho_j$.
 Any set of $\hat\pi_j$
satisfying Eqs. \rf{pih}, \rf{I}, and \rf{cond2} with
$\Gamma$ given in Eq. \rf{Gam} is a solution to this problem.
Proposition \ref{prop1} opens a way for
 an algorithmic solution to the problem.

 First, we note that the operators $\hat\pi_j=\o_j\pi_j$ are
 defined by the frequencies $\o_j$ and Bloch vectors
 $\vec\g_j=(\g_{j,1},\g_{j,2},\g_{j,3})$.
 According to Proposition
 \ref{prop1}, the vectors $\vec\g_j$ are defined
 by the parameter $A$ and vector $\vec B=(B_1,B_2,B_3)$.
 The still undefined parameters $A$, $\o_j$,
 and $\vec B$
 should be found form Eqs. \rf{gjk}, \rf{om1}, and \rf{om2}.
 Using Eq. \rf{gamma}, we obtain from Eq. \rf{gjk}
 a set of equations for  $B_k$ and $A$,
 \be{EB}
\frac14A^2-\frac12Ap_j=\sum_{k=1}^3B_k^2-
2p_j\sum_{k=1}^3B_k\b_{j,k}\,,\ j=1,\ldots,M\,.
 \ee
 It is convenient to rewrite this system as $M-1$
 homogeneous equations for $A$ and $B_k$ and an equation of
 the second order with respect to $A$ and $B_k$. For this
 purpose, we subtract from
 the Eq. \rf{EB} at $j=1$
the same equation at $j=2,\ldots,M$,
thus obtaining
 \be{EBA}
4\sum_{k=1}^3B_k(p_1\b_{1,k}-p_j\b_{j,k})=A(p_1-p_j)\,,
\quad j=2,\ldots,M
 \ee
and sum up the Eqs. \rf{EB} over $j$ from $1$ to $M$
which yields
\be{EBA2}
MA^2-2A\a=
4M\sum_{k=1}^3B_k^2-8\sum_{j=1}^M\sum_{k=1}^3B_kp_j\b_{j,k}\,,
\ee
where $\a=\sum_{j=1}^Mp_j$.
We thus have $M$ equations for four unknown parameters
$A$ and $B_{1,2,3}$.
Therefore, it is natural to expect that the system may have
solutions for $M=2,3,4$ and has no solutions for $M>4$
except for some special cases.
This means that except for some particular cases, the
optimal measurement for $N\ge4$ pure qubit states
is realized with four operators $\hat\pi_j$ or less.
Thus, if for $M=4$
the main determinant of the system \rf{EBA} is different
from zero, we can find all $B_j$ as linear functions of $A$.
The parameter $A=P_{corr}$ should be found from the
second-order Eq. \rf{EBA2}.
After that, we find $\gamma_{j,k}$ from Eq. \rf{gamma} and have
to check whether the systems \rf{om1} and \rf{om2} are
compatible.
If this is the case, from this system one can
find the frequencies $\o_j$.
 If all $\o_j$ are nonnegative, $\o_j\ge0$, the problem is
solved.
Otherwise, one has to choose $M=3$.
We have to note here that although in principle the
solution of the linear system
of equations \rf{om1}, \rf{om2} may always
be found analytically (for instance, as a ratio of corresponding
determinants), its explicit expression may be rather
involved (see some simple examples below) and difficult for analysis.
In such cases, numerical solution to this system may be
helpful.

For $M=3$, if we add to two equations \rf{EBA} the
compatibility condition of the system \rf{om1} considered
as a homogeneous system with respect to $\o_j$, we obtain an
inhomogeneous linear system of equations for $B_{1,2,3}$.
This system defines $B_j$ as linear functions of $A$ and
similar to the previous case, Eq. \rf{EBA2}
becomes a second-order equation with respect to $A$.
After that, we can omit
one of the equations \rf{om1} and from two remaining equations
and Eq. \rf{om2} find $\o_j$.
If all $\o_j\ge0$, the problem is
solved.
Otherwise, we have to choose $M=2$.
In this case, the
measurement becomes projective (i.e., von Neumann
measurement).
Since an optimal measurement exists always
(see e.g., \cite{Yuen}) in this case among $N$ given states, there always
exist two states (say $\rho_{1}$ and $\rho_{2}$) such that if we assign them the
probabilities $\frac{1}{\alpha}p_{1,2}$, $p_1+p_2=\a$, they
are optimally discriminated with an orthogonal measurement
and all the other states satisfy the conditions formulated
in the proposition \ref{prop0}.
This property follows from the invariance of the detection
operators with respect to scaling transformation
$p_j\to\a p_j$ mentioned in Sec. \ref{MinSubs}.
Below
for $N=2$ (Sec. \ref{N2ex}) this invariance will be
 demonstrated explicitly.

\subsection{Inverse optimization problem}

By inverse optimization problem, we mean a problem of
finding all possible states $\rho_j$ and
the prior probabilities $p_j$
for which a given
measurement strategy is optimal.
For qubit states, this means that given the frequencies
$\o_j$ and parameters $\g_{j,k}$,
we have to find $p_j$ and
$\b_{j,k}$, $j=1,\ldots,N$, $k=1,2,3$.
We formulate the solution to this problem as the following
proposition.
\begin{prop}\label{prop2}
Let the Bloch vectors $\vec\g_j$,
$j=1,\ldots,N$ and numbers $\o_j\ge0$
such that $\sum_{j=1}^N\o_j\vec\g_j=0$
and $\sum_{j=1}^N\o_j=2$
defining a POVM
$\hat\pi_j=\o_j\pi_j$
be given.
Then this POVM corresponds to an optimal measurement strategy for
the states $\rho_j$ given
with the probabilities
\be{pj}
p_j=Aq_j
\ee
by the Bloch vectors
\be{bjk}
\vec\b_{j}=\frac{(1-q_j)\vec\g_{j}+\vec R}{q_j}\,,
\quad j=1,\ldots,N\,,
\ee
where
\be{AFq}
q_j=\frac{\frac14+(2\vec\g_{j}+\vec R)\cdot\vec R}%
{\frac12+2\vec\g_{j}\cdot\vec R}\,,
\ee
\be{AF}
A=1/F\,,\quad
F=\sum_{j=1}^Nq_j\,,
\ee
where $\vec R$ is an arbitrary real vector
and by dot we denote the scalar product of the vectors
from a three-dimensional Euclidian space.
The optimal success probability is
$P_{corr}=A$.
\end{prop}

\begin{proof}
Some necessary conditions for an optimal measurement are
formulated in Proposition \ref{prop1}.
Thus, assuming the vectors $\vec\g_j$ be given in a frame by
their coordinates $\g_{j,k}$ and using Eq. \rf{gamma}, we find
the coordinates
$\b_{j,k}$
of the vectors $\vec\b_j$,
\be{betajk}
\b_{j,k}=\frac{(A-p_j)\g_{j,k}+B_k}{p_j}\,,
\ee
where $B_k$ are arbitrary real numbers which define a
vector $\vec B=(B_k)$.
For $\b_{j,k}$ to be coordinates of Bloch vectors, they
should satisfy Eq. \rf{bt0} which may be
solved with respect to $p_j$,
\be{pjA}
p_j=\frac{\frac14A^2+\sum_{k=1}^3B_k^2+2A\sum_{k=1}^3\g_{j,k}B_k}%
{\frac12A+2\sum_{k=1}^3\g_{j,k}B_k}\,.
\ee
If now we put
$\vec B=A\vec R$, i.e.,
$B_k=AR_k$, we obtain from here just Eqs.
\rf{pj} and \rf{AFq}.
At the same time, Eq. \rf{betajk}
reduces to Eq. \rf{bjk}.
Equation \rf{AF} follows from the condition
$\sum_{j=1}^Np_j=1$.
\end{proof}

\section{Particular cases}

\subsection{N=2\label{N2ex}}

Although the case of two qubit states is studied in details
\cite{Helstr,Barnett-Croke}, we find instructive to show how the known solution
follows from our approach.
Moreover, in view of
Proposition \ref{prop0}, we need to consider the case where
$p_1+p_2=\a<1$ which was not studied so far.

We find convenient to choose a frame in the three-dimensional
Euclidian space such that
the given states $\rho_1$ and $\rho_2$
are defined by the Bloch vectors
$\vec\b_1$ and $\vec\b_2$
 with coordinates
$\vec\b_1=(\b_{1,1},\b_{1,2},0)$ and
$\vec\b_2=(\b_{1,1},-\b_{1,2},0)$.

From Eq. \rf{om1} at $k=3$, we learn that $B_3=0$.
After solving the compatibility condition of two remaining
Eqs. \rf{om1} (with $k=1$ and $k=2$)
together with Eq. \rf{EBA}, we find $B_1$
and $B_2$ as linear functions of $A$,
\ba{B1}
B_{1}&=&\frac{\b_{1,1}}{4D^2}[A(p_1-p_2)^2+8\b_{1,2}^2p_1p_2(p_1+p_2)]\,,\\
B_{2}&=&\frac{\b_{1,2}(p_1-p_2)}{4D^2}[A(p_1+p_2)-8\b_{1,1}^2p_1p_2]\,,
\label{B2}
\ea
where
\be{D}
D=[\b_{1,1}^2(p_1-p_2)^2+\b_{1,2}^2(p_1+p_2)^2]^{1/2}\,.
\ee
Placing thus found $B_1$ and $B_2$ to Eq. \rf{EB}
at $j=1$, we solve
this equation with respect to $A$,
\be{A}
A=\frac12(p_1+p_2)+D\,.
\ee
Here from two roots of Eq. \rf{EB}, we kept the biggest one.
With the values $A$ and $B_{1,2}$ given in Eq. \rf{A} and
Eqs. \rf{B1} and \rf{B2},
one can transform the left-hand side of
Eq. \rf{om1} at $k=1$ to a form $(\o_2-\o_1)$ times a
nonzero factor so that one finds from the system \rf{om1}
and \rf{om2} that $\o_1=\o_2=1$.

Once the coordinates $B_{1,2,3}$ of the vector $\vec B$ and
the optimal success probability $P_{corr}=A$ \rf{A}
are found, we may calculate
detection operators
\be{pi1}
\pi_1=\left(
\begin{array}{cc}
\frac12 & \pi\\ \pi^{*} & \frac12
\end{array}
\right),\quad \pi_2=I-\pi_1\,,
\ee
\be{pi2ex}
\pi=\frac{1}{2D}[\b_{1,1}(p_1-p_2)-i\b_{1,2}(p_1+p_2)]
\ee
and the cost matrix
\be{GAM}
\Gamma=\frac14\left(
\begin{array}{cc}
2A & \gamma^*\\ \gamma & 2A
\end{array}
\right),
\ee
\ba{Gam12}
\gamma=\frac{1}{D}
\{i\b_{1,2}(p_1-p_2)(2D+p_1+p_2)
\\
+\b_{1,1}[(p_1-p_2)^2+2D(p_1+p_2)]\}\,.\nn
\ea
Here we observe the invariance of the detection operators
\rf{pi1} with respect to scaling
of the probabilities
$p_j\to\a p_j$
mentioned above.

If we notice that the parameters $\b_{1,1}$ and $\b_{1,2}$
may be expressed in terms of the Bloch vectors $\vec\b_1$
and $\vec\b_2$,
\be{beta12}
\b_{1,1}=\frac{(\vec\b_2+\vec\b_1)\cdot\vec\b_2}%
{\|\vec\b_2+\vec\b_1 \|}\,,\quad
\b_{1,2}=\frac{(\vec\b_2-\vec\b_1)\cdot\vec\b_2}%
{\|\vec\b_2-\vec\b_1 \|}\,,
\ee
where by
$\|\vec\b\|$ the length of the vector $\vec\b$ is denoted,
we may present parameter $D$ and with it the optimal
success probability $A$ \rf{A} in an explicitly covariant
 form
\be{Dcov}
D=\left[\frac{p_1^2+p_2^2}{4}-2p_1p_2\vec\b_1\cdot\vec\b_2\right]^{1/2}.
\ee
Another covariant form for $A$ corresponds to using the absolute value of
the overlap between the state vectors
$\la\psi_1|\psi_2\ra=s_{1,2}$,
$\rho_1=|\psi_1\ra\la\psi_1|$, $\rho_2=|\psi_2\ra\la\psi_2|$,
\be{s12}
|s_{1,2}|^2=\tr(\rho_1\rho_2)=
\frac12+2\vec\b_1\cdot\vec\b_2
\ee
which yields
\be{As12}
A=\frac{p_1+p_2}{2}+\frac12\left[(p_1+p_2)^2-4p_1p_2|s_{1,2}|^2\right]^{1/2}\,.
\ee
For $p_1+p_2=1$, the last equation gives
 the value known as the Helstrom bound
\cite{Helstr} (see also \cite{Barnett-Croke}).

Once the problem for two states given with the
probabilities $p_1$ and $p_2$, $p_1+p_2\le1$, is solved,
 we
can apply Proposition \ref{prop0} for formulating
conditions for $N>2$ qubit states which are optimally
discriminated by the same orthogonal measurement which
optimally distinguishes the given states $\rho_1$ and
$\rho_2$.
For that,
assuming that the states $\rho_j$ are given with the
probabilities $p_j$, we have to calculate
both the trace and the determinant of
the matrix $G_j=\Gamma-p_j\rho_j$.
After some algebra, one finds
\be{trG3}
\tr G_j=A-p_j\,,
\ee
\ba{detG3}
\mbox{det}\,G_j&=&\frac{D}{2}(A-p_j)-\frac{A(p_1-p_2)^2}{8D}
\nn\\
&+&
\frac{p_j(p_1-p_2)}{2D}
\left[ (p_1\vec\b_1-p_2\vec\b_2)\cdot\vec\b_j\right]
\nn\\
&+&{p_j}\left[p_1(\vec\b_1-\vec\b_j)+
p_2(\vec\b_2-\vec\b_j)\right]\cdot\vec\b_j\,.
\ea
If both $\tr G_j>0$ and $\mbox{det}\,G_j\ge0$,
$j=3,\ldots,N$,
then the optimal measurement is the projective measurement
which optimally distinguishes $\rho_1$ from $\rho_2$
given with the probabilities $p_1$ and $p_2$, $p_1+p_2<1$.

For $N$ equiprobable states,
given with the probability $p=1/N$,
we have $p_j=p$ and
$A=p+D$ so that  Eq. \rf{detG3} reduces to
\be{detG3Eq}
\mbox{det}\,G_j=p^2(\vec\b_1+\vec\b_2)\cdot(\vec\b_j-\vec\b_1)\,.
\ee

\subsection{$N=3$}

For three given states
$\rho_1$, $\rho_2$, and $\rho_3$,
it may happen that the optimal
measurement is the projective measurement.
In such a case, because of the aforementioned scale
invariance of the detection operators for $N=2$,
there are three possibilities
for the optimal discrimination among the states
$\rho_1$, $\rho_2$, and $\rho_3$.
The optimal measurement may be
 the projective measurement which optimally
discriminates among any two from three given states.
Therefore, to distinguish these possibilities below,
 we will need to put additional labels to quantities
$A$ and $G_j$.
In cases where there is no confusion, we will use the
previous notations also.
Thus, $A$ defined by Eqs. (\ref{A}) and \rf{Dcov} will be
denoted as $A_{1,2}$ so that
\be{A12}
A_{1,2}:=\frac12(p_1+p_2)+
\left[\frac{p_1^2+p_2^2}{4}-2p_1p_2\vec\b_1\cdot\vec\b_2\right]^{1/2}
\ee
and $\mbox{det}\,G_{j}$ given in Eq. \rf{detG3} at $j=3$
will be denoted $\mbox{det}\,G_{1,2;3}$, i.e.,
$\mbox{det}\,G_{1,2;3}:=\mbox{det}\,G_3$.
For the set of labels $\{1,2,3\}$, we will use also cyclic
permutations $\{1,2,3\}\to\{2,3,1\}\to\{3,1,2\}$.
Using these notations,
we may apply the proposition \ref{prop0} to formulate
conditions distinguishing projective optimal measurements from
generalized ones for discriminating among three qubit states.
\begin{prop}\label{prop3}
Let the different states $\rho_1$, $\rho_2$, and $\rho_3$ be given
with the probabilities $p_1$, $p_2$, $p_3$
and $\vec\b_1$, $\vec\b_2$,
$\vec\b_3$ as their Bloch vectors.
The optimal measurement discriminating among these states
is the projective measurement if the inequalities
$A_{i,j}>p_k$ and $G_{i,j;k}\ge0$ hold for one
cyclic permutation of the labels $\{i,j,k\}=\{1,2,3\}$
 at least.
Otherwise, i.e., if no cyclic permutation exists such that
both above inequalities hold, the optimal measurement is
generalized.
In the first case,
 the optimal success probability is $P_{corr}=A_{i,j}$.
 In the second case,
 $P_{corr}=A$ where $A$,
should be found from the system
of equations \rf{EB}, \rf{om1}, and
\rf{om2}.
\end{prop}
\begin{proof}
First we note that since the states are assumed to be
different from each other,
the optimal POVM contains more than one element.
The statement becomes evident if one notices that
the optimal solution exists always and
the optimal POVM contains three elements if it
cannot contain two elements.
\end{proof}

Note nevertheless that for the general case
of three qubit states,
solution to the
system of equations
\rf{EB}, \rf{om1}, and \rf{om2} is rather involved.
Therefore, below we will consider some particular cases.

\subsubsection{Equiprobable states}

Let us consider $p_1=p_2=p_3=p=1/N$ ($N=3$).
In this case,
$A=p\,\left(1+\sqrt{\frac12-2\vec\b_1\cdot\vec\b_2}\right)>p$
and whether the optimal measurement is projective or
generalized is defined by the sign of $\mbox{det}\,G_{1,2;3}$.
Thus, using Eq. \rf{detG3Eq} at $j=3$,
 we
may formulate a modification of Proposition \ref{prop3}
for this particular case.
\begin{prop}\label{prop3eq}
For three different equiprobable states
$\rho_j$
given by their Bloch vectors $\vec\b_j$, $j=1,2,3$,
the optimal POVM
includes three elements if
the following inequalities,
\ba{cond3a}
(\vec\b_1+\vec\b_2)\cdot(\vec\b_3-\vec\b_1)<0\,,
 \\ \label{cond3b}
(\vec\b_3+\vec\b_1)\cdot(\vec\b_2-\vec\b_3)<0\,,
\\
(\vec\b_2+\vec\b_3)\cdot(\vec\b_1-\vec\b_2)<0\,\,\,
 \label{cond3c}
\ea
hold.
Otherwise,
i.e., if at least one of these inequalities is violated
the optimal measurement is projective.
\end{prop}

If for instance $(\vec\b_1+\vec\b_2)\cdot(\vec\b_3-\vec\b_1)\ge0$
then the optimal measurement is the one which optimally
distinguishes the state $\rho_1$ from $\rho_2$.
In the coordinate system where
\ba{3csa}
\vec\b_1&=&(\b_{1,1},\b_{1,2},\b_{1,3})\,,\\
\label{3csb}
\vec\b_2&=&(\b_{1,1},-\b_{1,2},\b_{1,3})\,,\\
\label{3csc}
\vec\b_3&=&(\b_{3,1},\b_{3,2},\b_{1,3})\,,
\ea
\be{zz1}
\b_{1,1}\ge0\,,\quad \b_{1,2}>0\,,\quad \b_{1,3}\ge0
\ee
the optimal success probability reads
\be{A32eqp}
P_{corr}=A_{1,2}=\frac13(1+2\b_{1,2})\,.
\ee

For the generalized measurement,
 the system of Eqs. \rf{EB}, \rf{om1}, and \rf{om2}
 in the frame \rf{3csa}-\rf{3csc}
acquires a simple form and one easily obtains
that the vector $\vec B$ points to the positive direction
of the $z$ axis, $\vec B=(0,0,\frac13\b_{1,3})$,
and the optimal success probability depends on the common
projection of the vectors onto the same axis
\be{A3eqp}
P_{corr}=A=p\left(1+\sqrt{1-4\b_{1,3}^2}\right),\quad
p=\frac13\,.
\ee
For the frequencies $\o_j$,
 one gets
\be{o123peq}
\o_{1,2}=\frac{\mp\b_{1,1}\b_{3,2}-\b_{1,2}\b_{3,1}}%
{\b_{1,2}(\b_{1,1}-\b_{3,1})}\,,\quad
\o_{3}=\frac{2\b_{1,1}}{\b_{1,1}-\b_{3,1}}\,.
\ee
From the proposition \ref{prop3eq},
 it follows that
inequalities \rf{cond3a}-\rf{cond3c}
imply positivity of the frequencies \rf{o123peq}.

From Eq. \rf{gamma}, we find the vectors $\vec\g_j$,
$j=1,2,3$,
 defining
the detection operators $\hat\pi_j$,
$\vec\g_j=(1-4\b_{1,3}^2)^{-1/2}(\b_{j,1},\b_{j,2},0)$.
We thus see that these vectors
lay in the $z=0$ plane and differ from the projection of
the vectors $\vec\b_j$ onto the same plane by the factor
$(1-4\b_{1,3}^2)^{-1/2}$.
In particular, if the vectors $\vec\b_1$, $\vec\b_2$, and
$\vec\b_3$ are coplanar, i.e., if $\b_{1,3}=0$, the measurement
is generalized provided these vectors themselves may form a
POVM.

It is also useful to represent the success probability
\rf{A3eqp} in a covariant form
\be{A3cov}
P_{corr}=p\left(1+\sqrt{1-4\vec\b\cdot\vec\b}\right),\quad
p=\frac13\,,
\ee
\be{vecb}
\vec\b=\frac{[\vec\b_1\times\vec\b_2+\vec\b_2\times\vec\b_3+
\vec\b_3\times\vec\b_1]
(\vec\b_1,\vec\b_2,\vec\b_3)}%
{\|\b_1\times\b_2+\b_2\times\b_3+\b_3\times\b_1\|^2}
\ee
where by oblique cross, we denote
the vector product and by brackets
the mixed product of three-dimensional vectors.

Once the operators $\pi_j$ and frequencies
$\o_j$ are determined we may calculate the cost matrix
$\Gamma$. Thus from Eq. \rf{Gam}, it follows that
this matrix is diagonal
$\Gamma=\mbox{diag}(A/2+\b_{1,3},A/2-\b_{1,3})$.

Now,
assuming that the states $\rho_1$, $\rho_2$, and $\rho_3$
are discriminated by a generalized measurement and
using Proposition \ref{prop0}, we can formulate
conditions when the
optimal measurement for discriminating among
 $N\ge4$ equiprobable states given by their
Bloch vectors $\vec\b_j=(\b_{j,1},\b_{j,2},\b_{j,3})$,
$j=1,\ldots,N$,
is the same which is optimal for the first three states.
This is the case if both the trace and determinant of the
matrices $G_j=\Gamma-p\rho_j$, $j=4,\ldots,N$,
 are non-negative.
 Since the trace is always positive,
$\tr G_j=p\sqrt{1-4\b_{1,3}}>0$,
it remains to analyze the sign of the determinant
$\mbox{det}\,G_j=\b_{1,3}(\b_{j,3}-\b_{1,3})$.
Since $\b_{1,3}$ is assumed to be positive,
$\mbox{det}\,G_j>0$ if $\b_{j,3}-\b_{1,3}>0$.
From here, it follows that if all the additional vectors
in the chosen coordinate system
share the same hemisphere of the Bloch sphere and their
latitude is grater than or equal to the latitude of the three
given states,
 then they are discriminated by the same
measurement as the three given states and the optimal
success probability is given by Eq. \rf{A3eqp}
[or Eq. \rf{A3cov}] with $p=1/N$.
If the vectors are given in another coordinate system,
 the above condition takes the form
\be{detGcov}
\vec\b\cdot(\vec\b_j-\vec\b)\ge0\,,
\quad j=4,\ldots,N
\ee
with $\vec\b$ given in Eq. \rf{vecb}.

The optimization problem for three equiprobable states
plays a fundamental role since
for $N\ge4$ equiprobable states,
 the following statement takes place:
\begin{prop}\label{Nequalp}
Let $N\ge4$ different equiprobable states
$\rho_j$ be
given by their Bloch
vectors $\vec\b_j$, $j=1,\ldots,N$.
Then,
\newline
(i) if
no vector $\vec B\ne0$ exists such
that
\be{ntrB}
\vec B\cdot\vec\b_1=\vec B\cdot\vec\b_j\,,
\quad
j=2,\ldots,N
\ee
 and the vectors $\vec\b_j$, $j=1,\ldots,N$
cannot form a POVM,
i.e.,
no numbers $\o_j\ge0$, $j=1,\ldots,N$ such that
$\sum_{j=1}^N\o_j=2$ and $\sum_{j=1}^N\o_j\vec\b_j=0$ exist,
 then no solution to the optimization problem,
containing $N$ POVM elements, exists;
\newline
(ii)
if under the same assumption for the vector $\vec B$
 the vectors $\vec\b_j$ can form a POVM,
  then
$\vec\g_j=\vec\b_j$, $j=1,\ldots,N$, and
$P_{corr}=A=2p=2/N$;
\newline
(iii) if a vector $\vec B\ne0$
satisfying the conditions \rf{ntrB}
exists then the optimal
measurement may always be realized with either
two or three POVM elements.
\end{prop}

\begin{proof}
The statement (i) becomes evident in view of Eq. \rf{gamma}
which tells us that
 for $\vec B=0$,
  a solution is possible only if
$\vec\g_{j}=\vec\b_j$, in which case $\pi_j=\rho_j$.

The statement (ii) is a particular case of a solution
previously indicated by Yean et al. \cite{Yuen}.

Thus it remains to prove the statement (iii).
For that, we
will use the induction method and first prove the statement for
$N=4$.
If for $N=4$ there exists a vector $\vec B\ne0$
satisfying Eq. \rf{ntrB},
 then the system \rf{EBA} has a nontrivial solution
 and we easily find that
 in the coordinate system
 \rf{3csa}-\rf{3csc},
 the vector $\vec B$ points to the positive direction of
 the $z$ axis.
 Using Eqs. \rf{gamma}, \rf{om2} and \rf{om1}
 at $k=3$,
 we find its exact value, $\vec B=(0,0,p\,\b_{1,3})$,
 $p=1/4$.

 To find the frequencies $\o_j$, $j=1,\ldots,4$, we have
 to solve two from three
 equations given in Eq. \rf{om1} (at $k=1,2$)
 together with Eq.  \rf{om2}.
 From these equations one can always express say
 $\o_1$,  $\o_2$, and $\o_3$ in terms of $\o_4$ but
 only solutions with all $\o_k\ge0$, $k=1,\ldots,4$
 are suitable for our purpose.
 It is important to stress that
 $\o_{1,2,3}$ are linear (and hence monotonous) functions of
 $\o_4$.
 If we put $\o_4=0$,
  we find unique
 values for $\o_1$, $\o_2$, and $\o_3$
 since corresponding linear inhomogeneous system has a
 unique solution.
 If additionally all these
 frequencies are nonnegative, $\o_j\ge0$, $j=1,2,3$,
  then
 the optimal measurement contains three POVM elements
 $\pi_1$, $\pi_2$, and $\pi_3$ and the statement is
 proven.
 Let at $\o_4=0$ at least one of $\o_j$, $j=1,2,3$ is
 negative,
 say,
  for definiteness $\o_1<0$ and $\o_2\ge0$,
 $\o_3\ge0$.
 If the optimal measurement with four POVM elements exists,
then taking into account a monotonous dependence of
 $\o_{1,2,3}$ on $\o_4$,
  we conclude that
 there exists a value $\o_4=\o_{40}>0$ such
 that $\o_1=0$
 and $\o_2\ge0$, $\o_3\ge0$
 and the optimal
 measurement contains no more
 than three POVM elements
 $\pi_4$, $\pi_3$, and $\pi_2$
 which also tells us that the statement is proven.
 If the optimal
 measurement with four POVM elements does not exist,
  then
 it contains either three or two POVM elements
 since by assumptions that the given states are different
 from each other it cannot contain one element.
 Thus for $N=4$,
  the statement is proven.

Assume the statement be correct for $N$ Bloch vectors
 (the main assumption)
and
prove it for $N+1$ vectors.
For the general case
in the coordinate system \rf{3csa}-\rf{3csc}
from Eqs. \rf{EBA}, \rf{gamma}, \rf{om2}, and \rf{om1},
we
find the same value for the vector $\vec B$ as it was for
$N=4$ with the sole difference that now $p=1/(N+1)$.
To find $N+1$ frequencies $\o_j$,
 we have to solve
the same Eqs.
\rf{om1} at $k=1,2$ and Eq. \rf{om2}
from which we can always express three of them,
say, $\o_1$, $\o_2$, and $\o_3$ as linear
(and hence monotonous once again) functions
 of $\o_{j}$, $j=4,\ldots,N+1$.
Let us put $\o_j=0$, $j=4,\ldots,N+1$.
If we find $\o_{j}\ge0$, $j=1,2,3$,
 the statement is proven.
Therefore,
 assume that at least one of these frequencies is
negative, say, for definiteness $\o_1<0$ and $\o_2\ge0$,
$\o_3\ge0$.
The situation here is a bit more complicated than it was
for $N=4$ since the dependence of every $\o_j$, $j=1,2,3$ on the
other frequencies may be increasing with respect to some
frequencies and decreasing for some others.
If there exists a
measurement with $N+1$ POVM elements,
different scenarios may take place but in all cases,
 there
should exist a set of nonnegative frequencies.
In particular, it may happen that
there exists
such values of $\o_j=\o_{j0}\ge0$, $j=4,\ldots,N+1$ that
$\o_1=0$ and $\o_2\ge0$, $\o_3\ge0$.
Another possibility may be for instance such that at
$\o_1=0$,
 at least one of two other frequencies
becomes negative, say,
 $\o_2<0$, but with a variation of the
frequencies $\o_j$, $j=4,\ldots,N+1$
the frequency $\o_2$ should become
positive and hence should cross the point $\o_2=0$.
In other words, the set of the values which $\o_{1,2,3}$
take under the variation of $\o_j\in[0,2]$, $j=4,\ldots,N+1$
necessarily contains a point where at least one of $\o_1$,
$\o_2$, and $\o_3$ vanishes.
This means that the
optimal measurement should contain no more than $N$ POVM
elements.
For $\o_1=0$ and $\o_j\ge0$, $j=2$, \ldots, $N+1$,
 these are
$\pi_2$, \ldots, $\pi_{N+1}$.
Therefore, the matrix
$\Gamma=p\sum_{j=2}^{N+1}\o_j\pi_j\rho_j$
satisfy conditions \rf{cond1} and \rf{cond2} with
$p_j=p=1/(N+1)$.
Since $p$ is simply a scaling factor
both for $\Gamma$ and in Eq. \rf{cond2},
 these conditions
remain valid after the replacement $N+1\to N$. This means
that the same measurement is optimal for $N$ states
$\rho_2$, \ldots, $\rho_{N+1}$
 and according to the main assumption,
 there exists an optimal measurement which
contains either two or three POVM elements.
\end{proof}

From this proposition we can extract the following
corollary:
\begin{cor}
For $N\ge4$ equiprobable states
$\rho_j$ given by their Bloch vectors $\vec\b_j$,
 the optimal measurement may
contain $M\ge4$ POVM elements only if there exists a subset
of $M$ states which form a POVM.
\end{cor}

\begin{proof}
Indeed, if for $N$ states
only $\vec B=0$ satisfies
 Eqs. \rf{ntrB}
 and the states $\rho_j$ cannot form a POVM,
 then one has to examine all
$N$ subsets of $N-1$ states and check
whether there exists a subset which may form a POVM.
If the answer is positive, for instance the vectors
$\rho_j$, $j=1,\ldots,N-1$,
 form a POVM,
i.e., there exist
numbers $\o_j\ge0$
such that $\sum_{j=1}^{N-1}\o_j\rho_j=I$
or equivalently
such that $\sum_{j=1}^{N-1}\o_j=2$ and
$\sum_{j=1}^{N-1}\o_j\vec\b_{j}=0$,
then
as a POVM for $N$ states,
 we may choose just these
frequencies together with the states $\pi_j=\rho_j$,
$j=1,\ldots,N$ and put $\o_{N}=0$.
For this POVM,
 the cost matrix \rf{Gam} reads
$\G=p\sum_{j=1}^{N}\o_j\rho_j^2=p\sum_{j=1}^{N-1}\o_j\rho_j=pI$.
Therefore, for any state $\rho$ given by the Bloch vector
$\vec\b=(\b_1,\b_2,\b_3)$ for the trace and determinant
of the matrix $G=\G-p\rho$,
 one gets $\tr G=p$,
$\mbox{det}\,G=p^2(1/4-\b_1^2-\b_2^2-\b_3^2)=0$.
This result means that the chosen POVM is optimal for
discriminating among all $N$ states $\rho_j$.

If the answer is negative, i.e.,
if such a subset of $N-1$ states does not exist,
 one has to examine all subsets of $N-2$ states, etc.
From here, it is clear that if there exists a subset of $M\le N$
states which may form a POVM,
 then the optimal POVM for
discriminating among $N$ states consists of the POVM thus
found supplemented by other states
(for $M<N$) with zero frequencies.
If for a subset of the states
there exists a vector $\vec B\ne0$ as a solution to
the Eqs. \rf{ntrB}, then according to the point (iii)
of the proposition \ref{Nequalp}, the optimal measurement
cannot contain more than three POVM elements.
\end{proof}

\subsubsection{Two from three states are equally likely}

Consider the case when $p_1=p_2=p$ and $p_3=1-2p$,
$0<p\le1/2$.
Proposition \ref{prop3} permits us to indicate the cases
when three states given with the above probabilities are
discriminated by the same projective measurement as two
states given with equal probabilities.
It is convenient to choose a frame such that
\ba{zz2}\nn
\vec\b_1&=&(\b_{1,1},0,\b_{1,3})\,,\quad
\vec\b_2=(-\b_{1,1},0,\b_{1,3})\,,\nn \\
\b_{1,1}&>&0\,,\quad \b_{1,3}=\sqrt{1/4-\b_{1,1}^2}>0\,.\nn
\ea
From Eq. \rf{A12}, we find $A_{1,2}=p\,(1+2\b_{1,1})\ge1-2p$
or equivalently $p\ge(3+2\b_{1,1})^{-1}$.
Another condition follows from
nonnegativity of $\mbox{det}\,G_{1,2;3}$ given in Eq. \rf{detG3} at
$j=3$
\ba{z3}\nn
\mbox{det}G_{1,2;3}&=&p\b_{1,1}(3p+2p\b_{1,1}-1)
\\
&&
+
2p(2p-1)(1/4-\b_{1,3}\b_{3,3})\ge0\,.\nn
\ea
Since both $\b_{1,1}$ and $\b_{1,3}$ are
assumed to be positive and hence
$1+3\b_{1,1}+2\b_{1,1}^2-4\b_{1,3}\b_{3,3}>0$
from here,
it follows that
\be{z4}\nn
p\ge\frac{1+2\b_{1,1}-4\b_{1,3}\b_{3,3}}%
{2(1+3\b{1,1}+2\b_{1,1}^2-4\b_{1,3}\b_{3,3})}>
\frac{1}{3+2\b_{1,1}}
\ee
In particular at $\b_{3,3}=1/2$,
$\b_{1,1}=\frac12\sin(2\theta)$, and
$\b_{1,3}=\frac12\cos(2\theta)$,
$0\le\theta\le\pi/2$,
the last inequality takes the form
\be{z5}\nn
p\ge\frac{2}{5+\cos(2\theta)+\sin(2\theta)}>
\frac{1}{3+\sin(2\theta)}
\ee
and one recognizes the condition found in \cite{Anderson}
for
distinguishing the optimal projective measurement from the
generalized one for
discriminating among three mirror symmetric states.

Another simplification occurs at $\b_{3,3}=0$,
\be{z6}\nn
\mbox{det}\,G_{1,2;3}=\frac{p}{2}(1+2\b_{1,1})(2p+2p\b_{1,1}-1)\,,
\ee
 which
leads to
\be{z7}\nn
p\ge\frac{1}{2+2\b_{1,1}}=\frac{1}{2+\sin(2\theta)}>
\frac{1}{3+\sin(2\theta)}\,.
\ee
Note that this case was previously studied numerically in
\cite{Czhec} and we find instructive to compare our
analytic solution with the numerical one.
For this purpose, we additionally put
$\b_{3,2}=0$, $\b_{3,1}=\frac12$, and $\b_{1,1}=-\b_{2,1}=\frac{b}{2}$,
$\b_{1,3}=\b_{2,3}=\frac12\sqrt{1-b^2}$, $b>0$.
Solving simultaneously Eqs. \rf{EBA2} and the compatibility
condition for the system \rf{om1},
 we find the vector
\be{z8}\nn
\vec
B=\left(\frac{A(3p-1)}{2p\sqrt{1-b^2}}\,,0\,,0\right),
\ee
which being put to Eq. \rf{EBA2} with $\a=1$
gives a second-order
equation with respect to $A$.
This equation has the roots $A=0$ and
\be{z9}\nn
A=\frac{2p^2(1-2p)(b^2-1)}{1-6p+p^2(8+b^2)}\,.
\ee
With these values of $B_j$ and $A$ we solve the system
of equations \rf{om1}, \rf{om2}
\ba{om123ch}
\o_{1,2}&=&\pm\frac{1+p[4(p-1)+b^2(5p-2)]}%
{b[1+p(-6+(8+b^2)p)]^2}
\nn \\
&\times&[-1+p(4-4p+b[\mp2+(\pm6+b)p])]\,,
\\
\o_{3}&=&2-\o_1-\o_2\,.
\ea
From here, it follows that the equation $\o_3=0$ has only one
real root inside the interval $0<p<1/2$, which is
$p=p_{r}=(2+b)^{-1}=(2+2\b_{1,1})^{-1}$ and if
$p>p_{r}$, $\o_3<0$.
The equation $\o_2=0$ has no real roots inside this interval
and the equation $\o_1=0$ has one root also
$p=p_{l}=(b-2+\sqrt{2b(1+b)})/(b^2+6b-4)$.
Thus, for $p_l<p<p_r$, these states are discriminated by a
generalized measurement.
Otherwise,
 they are discriminated by a
projective measurement.
In contrast to the paper \cite{Czhec},
 we give exact analytic
bounds for the probability $p$ distinguishing
optimal generalized
measurements from projective ones.

\subsection{N=4}

Let two states $\rho_1$ and $\rho_2$ be given with the
probability $p$ and two other states $\rho_3$ and $\rho_4$
with the probability $1/2-p$, $0<p<1/2$.
Let us choose the
coordinate system such that the
$z=0$ plane is the plane of the
 Bloch vectors $\vec\b_1$, $\vec\b_2$ and they are
 symmetrically disposed with respect to $x$ axis, i.e.,
 \be{z10}\nn
 \vec\b_1=(\b_{1,1},\b_{1,2},0)\,,\quad
\vec\b_2=(\b_{1,1},-\b_{1,2},0)\,,
\ee
and assume that $\b_{1,1}\ge0$ and $\b_{1,2}>0$.
Here,
we will consider the case where the vectors $\vec\b_3$
and $\vec\b_4$ are placed into the $x=0$ plane and
symmetrically disposed with respect to $z$ axis,
i.e.,
 \be{z11}\nn
 \vec\b_3=(0,\b_{3,2},\b_{3,3})\,,\quad
\vec\b_2=(0,\b_{3,2},-\b_{3,3})
\ee
with $\b_{3,2}\ge0$ and $\b_{3,3}>0$.

From the system \rf{EB}, we find both the vector
$\vec B=(B_1,0,0)$
and the optimal success probability $P_{corr}=A$,
\be{Aex4}
B_1=\frac{2p\b_{1,1}(1-6p+8p^2)}{1-8p+16p^2(1-\b_{1,1}^2)}\,,
\quad
A=\frac{8p\b_{1,1}}{4p-1}B_1\,,
\ee
if for
discrimination of these states,
 one needs four detection
operators

From here,
 we see that $P_{corr}$ does not depend on the
states $\rho_3$ and $\rho_4$.
In particular, these states may be orthogonal to each
other ($\b_{3,2}=0$) in which case they
correspond
to opposite
points of the Bloch sphere.
With these values of $B_1$, $B_2$, $B_3$, and $A$, from the system
of equations \rf{om1} and \rf{om2}, we find the frequencies
\ba{o12ex4}\nn
\o_{1,2}=\pm\frac{1+4p\,[p-1+4(3p-1)\b_{1,2}^2]}%
{\b_{1,2}[1+4p\,(-2+p\,[3+4\b_{1,2}^2])]^2}\,\,\\
\times\left[
\pm\,4p\,(4p-1)\b_{1,2}+(-1+4p\,[1-4p\,\b_{1,1}^2]\,)\b_{3,2}
\right],\,\,
\ea
\be{o3ex4}
\o_{3}=\o_4=1-\frac12(\o_1+\o_2).
\ee
Note that for $\b_{3,2}=0$,
when the states $\rho_3$ and $\rho_4$ are orthogonal to
each other,
 $\o_1=\o_2$.
From here it follows that for a fixed value of $\b_{1,1}$
(and hence $\b_{1,2}$), $\o_4$ is not positive,
 $\o_4\le0$, if
\be{z12}\nn
p\ge\frac{1-2\b_{1,2}}{8\b_{1,1}^2}:=p_r
\ee
 and
$\o_1>0$ if
\be{z13}\nn
p>\frac{\b_{1,2}-\b_{3,2}+
\sqrt{\b_{1,2}(\b_{1,2}+2\b_{3,2}+4\b_{1,2}\b_{3,2}^2)}}%
{8\b_{1,2}-2\b_{3,2}+8\b_{1,2}^2\b_{3,2}}:=p_{\,l}
\ee
whereas $\o_2>0$ inside the interval $p_l<p<p_r$.
Note that for $\b_{1,2},\b_{3,2}\in[0,1/2]$,
 the difference
$p_r-p_l$ is not negative, $p_r-p_{\,l}\ge0$.
Thus,
for the $p$ values from the interval $p\in(p_l,p_r)$,
 the states are discriminated with four
 detection operators
 $\hat\pi_j=\o_j\pi_j$, $j=1,\ldots,4$.
 For $p\ge p_r$, they are discriminated by a projective
 measurement which optimally distinguishes $\rho_1$ from
 $\rho_2$.
 For the success probability,
  one finds from Eq. \rf{Aex4} at $p=p_r$
 the value $P_{corr}=1/2$.
 Note that Eqs. \rf{s12} and \rf{As12} at $p_1=p_2=p_r$
yield the same value for the success probability.
With $p$ growing from $p_r$ up to $1/2$, the success
probability linearly increases until the value
$1/2+\b_{1,2}$.

At $p=p_{\,l}$, $\o_1$ vanishes, $\o_1=0$, and
for $p\le p_{\,l}$,
the states are optimally
discriminated by the same measurement which optimally
discriminates among the states
$\rho_2$, $\rho_3$, and $\rho_4$.
From Eqs. \rf{o12ex4} and \rf{o3ex4} at $p=p_l$,
 we obtain corresponding frequencies
\be{o34}
\o_{3}=\o_4=\frac{\b_{1,2}+\b_{3,2}}%
{\b_{1,2}+2\b_{3,2}+4\b_{1,2}\b_{3,2}^2}\,,\quad
\o_2=2-2\o_3\,.
\ee
Note that for  $0<\b_{1,2},\b_{3,2}<1/2$ all the frequencies
\rf{o34} are positive.
On the other hand, for $\b_{3,2}=0$, when
the state $\rho_3$ is orthogonal to $\rho_4$,
$p_{\,l}=1/4$,
$\o_3=\o_4=1$, and $\o_1=\o_2=0$.
The success probability in this case is $P_{corr}=A=1/2$.
Note also that the same value follows from Eqs.
\rf{A} and \rf{Dcov} at
$p_1=p_2=1/4$ and $\vec\b_1\cdot\vec\b_2=-1/4$.
When $p\to0$,
the success probability tends linearly to zero.


\section{Conclusion}

The known necessary and sufficient conditions for
discriminating between nonorthogonal qubit states
with the minimum error
are
formulated in terms of Bloch vectors representing the
states.
This permitted us to indicate an algorithmic
solution to the direct optimization problem and give a
complete solution to the inverse optimization problem.
By the direct optimization problem,
 we mean the problem of
finding the optimal measurement strategy when the states to
be discriminated are given together with their prior
probabilities.
Accordingly, the inverse
optimization problem is the problem of specifying all
possible states and their prior probabilities which may be
optimally discriminated by the given generalized
measurement.

An intermediate (or mixed)
optimization problem may also be important
for practical usage. We formulate it as follows.
Assume that the minimum-error discrimination strategy for $N$
given states is known.
Is it possible to enlarge this set
of states in such a way that they are optimally
discriminated by the same measurement strategy?
If yes, what are
 the conditions to which the additional states should satisfy?
A simple analysis of the necessary and sufficient
conditions for minimizing the generalized rate error
tells us that the answer to this question is positive.
We formulated conditions in terms of Bloch vectors when
a set of $N$ qubit states is discriminated by the
same orthogonal measurement which is optimal for
discriminating among two states given with arbitrary prior
probabilities.
We applied this result to formulate  conditions when
three qubit states are optimally discriminated either by the
generalized measurement or by the projective one.
In particular for $N$ equiprobable states,
 we have shown
that the optimal generalized measurement may contain
$M\ge4$ POVM elements only if the states themselves may
form a POVM.
Otherwise, it contains either three or two POVM
elements.
For three equiprobable states, we found conditions which
distinguish projective optimal measurements from
generalized ones.

To illustrate our approach,
 we considered an example of three states,
 two of which are equally likely.
 For a
particular case of three mirror symmetric states,
 we reproduced a result previously published in
\cite{Anderson}.
To illustrate the advantage of the analytic solution
compared to a numerical one,
we considered another particular case
previously reported in \cite{Czhec}.
For this case in contrast to \cite{Czhec},
we were able to present the strict analytic bounds
for the prior probability for distinguishing the optimal
generalized measurement from the projective one.
As the final illustration of our approach ,
we solved
the optimization problem for four qubit states,
 two of
which are given with the probability $p$ and two others
with the probability $1/2-p$.

\section*{Acknowledgments}

The work is partially supported
by President of Russia under Grant No. SS-871.2008.2,
 Russian
Science and Innovations Federal Agency under contract No.
02.740.11.0238,
and Russian Special Federal Program
under contract No. P-1337.
The author is grateful to J. Tyson for a useful
communication.


\begin{thebibliography}{99}

\bibitem{BHH}J. A. Bergou, U. Herzog, and M. Hillery,
Lect. Notes
Phys. {\bf 649},  417 (2004).

\bibitem{Barnett-Croke}S. M. Barnett and S. Croke,
Adv. Opt. Photonics. {\bf 1}, 238 (2009).

\bibitem{Helstr}C. W. Helstrom, {\it Quantum detection and
estimation theory}, (Academic, New York, 1976).

\bibitem{Chefless}A. Chefles, Contemp. Phys. {\bf 41},
401 (2000).

\bibitem{Hol}A. S. Holevo, J. Multivariate Anal. {\bf 3},
337 (1973).

\bibitem{Hol2}A. S. Holevo, Probl. Inf. Transm. {\bf 10},
317, (1974); [Probl. Peredachi Inf, {\bf 10}, 51 (1974)].

\bibitem{Yuen}
H. P. Yuen, R. S. Kennedy, and M. Lax,
Proc. IEEE {\bf 58}, 1770 (1970);
H. P. Yuen, R. S. Kennedy, and M. Lax,
IEEE Trans. Inf. Theory. {\bf 21}, 125 (1975).

\bibitem{Barnett-Croke-JPA}S. M. Barnett and S. Croke,
J. Phys. A: Math. Theor. {\bf 42}, 062001 (2009)


\bibitem{Hunter}K. Hunter, AIP Conf. Proc. {\bf 734}, 83
(2004).

\bibitem{Anderson}E. Andersson, S. M. Barnett, C. R.
Gilson, and K. Hunter, Phys. Rev. A {\bf 65}, 052308 (2002).

\bibitem{Czhec}M. Je\v{z}ek, J. \v{R}eh\'{a}\v{c}ek, and J. Fiur\'{a}\v{s}ek,
Phys. Rev. A {\bf 65}, 060301 (2002).

\bibitem{GM}A. S. Holevo, {\it Statistical structure of quantum theory}
(Springer-Verlag, Berlin,2001);
{\it Probabilistic and statistical aspects of
quantum theory} (Nauka, Moscow, 1980);
A. Peres, {\it Quantum Theory: Concepts and
Methods} (Kluwer, Dordrecht, 1995).


\end{thebibliography}
\end{document}